\documentclass[letterpaper, 10 pt, conference]{ieeeconf}  % Comment this line out if you need a4paper

\usepackage{graphicx}
\usepackage{booktabs} % for professional tables

\usepackage{hyperref}       % hyperlinks
\usepackage{url}            % simple URL typesetting
\usepackage{amsfonts}       % blackboard math symbols
\usepackage{amsmath, amssymb,  enumerate, bm, capt-of,mathtools,nicefrac}
\usepackage[capitalise]{cleveref}
\usepackage{siunitx}
\usepackage{xcolor}
\usepackage[norelsize,ruled,linesnumbered,noend]{algorithm2e}
\usepackage{setspace}

\usepackage{tikz, tikzscale,pgfplots}
\usetikzlibrary{patterns}
\usepgfplotslibrary{fillbetween}

\pgfplotsset{width=10\columnwidth /10, compat = 1.13, 
	height = 55\columnwidth /100, grid= major, 
	legend cell align = left, ticklabel style = {font=\scriptsize},
	every axis label/.append style={font=\scriptsize},
	legend style = {font=\tiny},title style={yshift=-7pt, font = \small} }

\newcommand{\changed}[1]{{\color{black} #1}}

\newtheorem{assumption}{Assumption}
\newtheorem{definition}{Definition}
\newtheorem{lemma}{Lemma}
\newtheorem{proposition}{Proposition}
\newtheorem{theorem}{Theorem}
\newtheorem{remark}{Remark}

\DeclareMathOperator*{\argmin}{arg\,min}
\DeclareMathOperator*{\argmax}{arg\,max}

\IEEEoverridecommandlockouts                              % This command is only needed if 
\usepackage{background}
\SetBgContents{Copyright (c) 2023 IEEE}
\SetBgScale{1}
\SetBgAngle{0}
\SetBgPosition{current page.north east}
\SetBgHshift{-2.5cm}
\SetBgVshift{-1cm}
\title{\LARGE \bf
Risk-Sensitive Inhibitory Control for Safe Reinforcement Learning
}

\author{Armin Lederer$^{1}$, Erfaun Noorani$^{2}$, John S. Baras$^{2}$, Sandra Hirche$^{1}$% <-this % stops a space
\thanks{$^{1}$A. Lederer and S. Hirche are with the Chair of Information-oriented Control (ITR), School of Computation, Information and Technology, Technical University of Munich, 80333 Munich, Germany. Emails: {\tt\small \{armin.lederer, hirche\}@tum.de}.}%
\thanks{$^{2}$E. Noorani and J. Baras are with the Department of Electrical and Computer Engineering and the Institute for Systems Research (ISR) at the University of Maryland, College Park, MD, USA. Emails: {\tt\small \{enoorani,baras\}@umd.edu}.}%
\thanks{Research partially supported by ONR grant N00014-17-1-2622, by a grant from the Army Research Lab, by the Clark Foundation, and by the European Research Council (ERC) Consolidator Grant 
"Safe data-driven control for human-centric systems (CO-MAN)" under grant 
agreement number 864686.}
}

\begin{document}

\maketitle
\thispagestyle{empty}
\pagestyle{empty}

%%%%%%%%%%%%%%%%%%%%%%%%%%%%%%%%%%%%%%%%%%%%%%%%%%%%%%%%%%%%%%%%%%%%%%%%%%%%%%%%
\begin{abstract}
Humans have the ability to deviate from their natural behavior when necessary, which is a cognitive process called response inhibition. 
Similar approaches have independently received increasing attention in recent years for ensuring the safety of control. Realized using control barrier functions or predictive safety filters, these approaches can effectively ensure the satisfaction of state constraints through an online adaptation of nominal control laws, e.g., obtained through reinforcement learning. While the focus of these realizations of inhibitory control has been on risk-neutral formulations, human studies have shown a tight link between response inhibition and risk attitude. Inspired by this insight, we propose a flexible, risk-sensitive method for inhibitory control. Our method is based on a risk-aware condition for value functions, which guarantees the satisfaction of state constraints. We propose a method for learning these value functions using common techniques from reinforcement learning and derive sufficient conditions for its success. By enforcing the derived safety conditions online using the learned value function, risk-sensitive inhibitory control is effectively achieved. 
The effectiveness of the developed control scheme is demonstrated in simulations.
\end{abstract}

%%%%%%%%%%%%%%%%%%%%%%%%%%%%%%%%%%%%%%%%%%%%%%%%%%%%%%%%%%%%%%%%%%%%%%%%%%%%%%%%
\section{Introduction}

\setlength{\textfloatsep}{8pt}
\setlength{\abovedisplayskip}{6pt}
\setlength{\belowdisplayskip}{6pt}

Having a pause before responding is a mental technique that helps humans perceive, control, and manage our emotions. Human's ability to think before reacting, especially in difficult and complex situations, is a cognitive mechanism to keep our actions in check. This cognitive process is called inhibitory control, also known as response inhibition~\cite{Nigg2000}. Response inhibition allows an individual to inhibit their prepotent (natural and habitual) responses in order to select a more appropriate (e.g. safer) behavior. 

Independent from this foundation in psychology, response inhibition has become increasingly popular in learning-based control \cite{Brunke2021} and Reinforcement Learning (RL) \cite{Sutton2017} in recent years, where safety is a major concern \cite{Dulac-Arnold2019}. The idea is to decouple optimality and safety by independently determining safe and optimal control laws. Before applying an optimal, but potentially unsafe control input to the real system, its safety is checked, such that a safe control input can be chosen instead \cite{Alshiekh2018}. Thereby, the prepotent optimal response is inhibited to guarantee the safety of the closed-loop system.%\looseness=-1

The challenge of this approach lies in finding safe policies and efficient methods to determine the safety of a control input online. When the dynamics of the systems are known to exhibit a control-affine structure, control barrier functions (CBF) can be effectively employed to address this challenge~\cite{Taylor2019}. Since their analytical derivation for more flexible classes of dynamical systems is difficult at best, techniques from model predictive control have become popular for computing safe backup strategies online~\cite{Bastani2021, Wabersich2021b}. While such predictive safety filters provide a conceptionally flexible approach for realizing inhibitory control, they generally suffer from high computational complexity. 
% and the propagation of uncertainty due to noise is difficult. 
This limitation can be mitigated by combining ideas from reachability analysis~\cite{Hsu2021} or optimal control~\cite{Curi2022} with reinforcement learning techniques to learn safety conditions and safe control laws offline, such that resource-demanding computations can be avoided during the application of the inhibited control law.

While these approaches allow the seemingly straightforward realization of inhibitory control for ensuring the safety of real-world systems, they do not consider the risk of losing safety 
due to uncertainty arising from approximate system models and process noise. This is in strong contrast to humans, for which psychological studies have shown a critical link between response inhibition and an individual’s risk attitude (willingness to take risk or not)~\cite{sherman2018connecting}. When inhibitory control is implemented in technical systems through analytically derived safety conditions such as CBFs, this risk-sensitivity can be easily achieved by reformulating standard conditions using risk measures~\cite{Ahmadi2022}. %, Singletary2023
However, the extension to flexible approaches for constructing safety conditions, e.g., using RL techniques remains an open problem.\looseness=-1

We address this problem of realizing inhibitory control with risk-awareness similar to humans for ensuring the safety of a wide class of systems via the following contributions:\looseness=-1
\begin{itemize}
    \item \textbf{Risk-sensitive safety conditions:} 
    To ensure the probabilistic satisfaction of state constraints, we introduce cost functions allowing us to express safety via risk-sensitive conditions on the cumulative cost along system trajectories. These conditions reveal an intuitive relationship between risk-aversion and safety probability.%\looseness=-1
    \item \textbf{Safe policies and value functions through RL:} Based on these results, we develop an approach for determining safe policies and corresponding safety value functions using common techniques from reinforcement learning. The success of the proposed approach is shown to be guaranteed under weak assumptions relating to the controllability properties of the system dynamics.
    \item \textbf{Inhibitory control through safety filters:} By enforcing the satisfaction of the derived safety conditions with the learned value function online, we obtain a risk-sensitive safety filter. %\textcolor{blue}{that is inspired by response inhibition in humans.} 
    Moreover, we prove it to inherit probabilistic safety guarantees from the safe policy obtained through RL.\looseness=-1
\end{itemize}

The remainder of this paper is structured as follows. In \cref{sec:problem}, the problem of rendering a given policy safe with respect to state constraints using safety filters is formalized. Our approach for realizing response inhibition in control using risk-sensitive safety filters is derived in  \cref{sec:inhibitory_control}. In \cref{sec:num_eval}, the effectiveness of the proposed safety filter is demonstrated, before the paper is concluded in \cref{sec:conclusion}.

\section{Problem Statement}\label{sec:problem}
We consider a discrete-time dynamical system\footnote{Notation: 
Lower/upper case bold symbols denote vectors/matrices, blackboard bold letters 
denote sets, 
$\mathbb{R}_+$/$\mathbb{R}_{0,+}$ all real positive/non-negative 
numbers, 
$\|\cdot\|$ the Euclidean norm, $\mathbb{E}_{x}[\cdot]$ the expectation with respect to the distribution of $x$, and $\mathbb{P}(\cdot)$ the probability.}
\begin{align}\label{eq:true_sys}
    \bm{x}_{k+1}=\bm{f}(\bm{x}_k,\bm{u}_k,\bm{\omega}_k),
\end{align}
where $\bm{x}_k\in\mathbb{X}\subset\mathbb{R}^{d_x}$ are states, $\bm{u}_k\in\mathbb{U}\subset\mathbb{R}^{d_u}$ are control inputs, $\bm{\omega}_k\in\Omega\subset\mathbb{R}^{d_\omega}$, \changed{$\bm{\omega}_k\sim\rho(\bm{x}_k)$ is independent process noise drawn from a potentially state-dependent distribution $\rho(\bm{x}_k)$ with zero mean,} and
$\bm{f}:\mathbb{X}\times\mathbb{U}\times\Omega\rightarrow\mathbb{X}$ denotes an unknown, continuous transition function. We assume that a nominal, potentially unsafe policy $\bm{\pi}^*:\mathbb{X}\rightarrow\mathbb{U}$ is given, which can be obtained, e.g., using standard reinforcement learning techniques \cite{Sutton2017}. 

The goal is to render the nominal policy safe using inhibitory control of the form
\begin{subequations}
    \begin{align}
    \bm{\pi}_{\mathrm{safe}}^*(\bm{x})=&\argmin\limits_{\bm{u}\in\mathbb{U}} \|\bm{\pi}^*(\bm{x})-\bm{u}\|\\
    &\text{such that $\bm{u}$ is safe}.\label{eq:safety_cond_problem}
\end{align}
\end{subequations}

In this response inhibition, our notion of safety follows the common principle of classifying the state space $\mathbb{X}$ into a safe region $\mathbb{X}_{\mathrm{safe}}\subset\mathbb{X}$ and an unsafe region $\mathbb{X}_{\mathrm{unsafe}}=\mathbb{X}\setminus\mathbb{X}_{\mathrm{safe}}$. For example, the safe set $\mathbb{X}_{\mathrm{safe}}$ can represent the joint angles for which self-collisions of a robotic manipulator are excluded. Due to the process noise $\bm{\omega}$ with a potentially unbounded probability distribution, it is generally not possible to deterministically ensure that the system never enters the unsafe state space $\mathbb{X}_{\mathrm{unsafe}}$. Therefore, we define safety probabilistically through the following form of forward invariance.
\begin{definition}\label{def:safety}
A policy $\bm{\pi}(\cdot)$ is called $\delta$-safe if there exists a subset $\mathbb{V}\subseteq\mathbb{X}_{\mathrm{safe}}$ such that $\mathcal{P}(\bm{f}(\bm{x},\bm{\pi}(\bm{x}),\bm{\omega})\in\mathbb{V})\geq 1-\delta$ for all $\bm{x}\in\mathbb{V}$.
\end{definition}
Since \cref{def:safety}  requires a form of forward invariance of $\mathbb{V}$, it immediately induces guarantees for all states along a $K$-step trajectories of the form 
\begin{align}
    \mathcal{P}(\bm{x}_k\in\mathbb{V},~ \forall k=1\ldots,K)\geq (1-\delta)^K,
\end{align}
where $\bm{x}_k$ is defined through iterative application of \eqref{eq:true_sys}. Hence, the considered notion of safety in this paper is stronger than merely requiring the next state to lie in the safe \changed{subset}, i.e., $\mathcal{P}(\bm{f}(\bm{x},\bm{\pi}(\bm{x}),\bm{\omega})\in\mathbb{X}_{\mathrm{safe}})\geq 1-\delta$.

Based on the definition of $\delta$-safety, we consider the problem of deriving a tractable safety condition \eqref{eq:safety_cond_problem} for inhibitory control, which is guaranteed to be feasible for some risk-aversion as measured through $\delta$. 
Since we assume the transition function $\bm{f}$ is unknown, solving this problem is generally impossible without any further assumptions. Therefore, we require the availability of a probabilistic model in the form of a distribution over functions as formalized in the following.\looseness=-1
\begin{assumption}\label{ass:prob_model}
A probability distribution $\mathcal{F}$ over potential dynamics $\bm{f}$ is known, i.e., $\bm{f}\sim\mathcal{F}$.
% The unknown function is a sample from a known probability distribution $\mathcal{F}$, i.e., $\bm{f}\sim\mathcal{F}$.
\end{assumption}
In practice, suitable distributions over functions $\mathcal{F}$ can be straightforwardly obtained using Bayes' theorem, e.g., through %linear Bayesian regression \cite{Bishop2006}, 
Gaussian process regression \cite{Rasmussen2006}. Moreover, approximate distributions can be learned using %Bayesian neural networks \cite{Depeweg2017} or 
deep ensembles~\cite{Lakshminarayanan2017}. Therefore, this assumption is not restrictive in practice.

%%%%%%%%%%%%%%%%%%%%%%%%%%%%%%%%%%%%%%%%%%%%%%%%%%
\section{Risk-Sensitive Inhibitory Control}\label{sec:inhibitory_control}

Even with the knowledge of $\mathcal{F}$, determining a safety condition \eqref{eq:safety_cond_problem} is a challenging problem since we 
generally do not know which subset $\mathbb{V}$ is suitable for \cref{def:safety}. Here, we follow the ideas of~\cite{Curi2022} and employ RL techniques to define these subsets through a value function. For this purpose, we first show how state constraints can be expressed through risk-sensitive cost conditions in \cref{subsec:cons2cost}.
After deriving these safety conditions, in \cref{subsec:safe_backup}, we address the problem of learning a separate, so-called backup policy whose pure focus lies on ensuring safety. Based on this policy, a risk-sensitive safety filter for realizing inhibitory control in reinforcement learning is finally presented in \cref{subsec:safety_filt}.\looseness=-1

%%%%%%%%%%%%%%%%
\subsection{State Constraints as Risk-Sensitive Cost Conditions}\label{subsec:cons2cost}

In order to express state constraints through risk-sensitive cost conditions, we define the expected
cumulative cost for a policy $\bm{\pi}(\cdot)$  as
\begin{align}\label{eq:cum_cost}
    V_{\bm{\pi}}(\bm{x}) =\mathbb{E}_{\bm{f},\bm{\omega}}\left[\sum\limits_{k=0}^{\infty}  \gamma^k c(\bm{x}_k)\right],
\end{align}
where $c:\mathbb{R}^{d_x}\rightarrow\mathbb{R}_{0,+}$ denotes an immediate cost, \changed{$\gamma\in(0,1)$ is a discount factor,} and $\bm{x}_k$ is defined through the iterative application of \eqref{eq:true_sys} with $\bm{x}_0=\bm{x}$ and $\bm{u}_k=\bm{\pi}(\bm{x}_k)$. If the immediate cost $c(\cdot)$ can be used as an indicator of the unsafe subset $\mathbb{X}_{\mathrm{unsafe}}$, there exists a sub-level set of $V_{\bm{\pi}}(\cdot)$ contained in $\mathbb{X}_{\mathrm{safe}}$, as guaranteed by the following lemma.
\begin{lemma}[\cite{Curi2022}]\label{lem:set2cost}
Assume there exists a constant $\hat{c}\in\mathbb{R}_+$, such that the cost $c:\mathbb{R}^{d_x}\rightarrow\mathbb{R}_{0,+}$ satisfies
\begin{align}\label{eq:c_cond}
    c(\bm{x})\geq \hat{c}\quad \forall\bm{x}\in\mathbb{X}_{\mathrm{unsafe}}.
\end{align} 
Then, there exists a constant $\bar{\xi}\in\mathbb{R}_+$, such that the intersection between the sub-level set $\mathbb{V}_{\bm{\pi}}^{\bar{\xi}}=\{\bm{x}\in\mathbb{X}: V_{\bm{\pi}}(\bm{x})\leq\bar{\xi}\}$ and $\mathbb{X}_{\mathrm{unssafe}}$ is empty, i.e., $\mathbb{V}_{\bm{\pi}}^{\bar{\xi}}\cap\mathbb{X}_{\mathrm{unsafe}}=\emptyset$. 
\end{lemma}
Based on this lemma, we can choose any sub-level set $\mathbb{V}_{\bm{\pi}}^{\xi}$ with $\xi\leq \bar{\xi}$ for showing $\delta$-safety as introduced in \cref{def:safety}. 
\changed{As discussed in \cite{Curi2022}, the immediate cost $c(\cdot)$ for defining sub-level sets $\mathbb{V}_{\bm{\pi}}^{\xi}$ can be selected relatively freely, such that simple choices as the indicator function are applicable in principle. However, this choice does not provide informative gradients, which complicates the learning process. Therefore, other cost functions such as rectified linear unit functions generally need to be considered, even though they can potentially lead to more conservative approximations of the safe set $\mathbb{X}_{\mathrm{safe}}$. 
}
\changed{To obtain suitable values for $\bar{\xi}$, different approaches can be used. For example, potentially conservative closed-form expressions can be employed as shown in \cite{Curi2022}. Moreover, optimal solutions can be found by formulating the search for $\bar{\xi}$ as a robust optimization problem, which can be solved numerically. }
Therefore, it only remains to derive conditions that ensure the state stays in $\mathbb{V}_{\bm{\pi}}^{\xi}$ after a transition. While this could be achieved using a probabilistic ''worst case'' consideration as shown in \cite{Curi2022}, this approach yields a computationally challenging min-max problem for unknown system dynamics. Therefore, we follow a fully probabilistic approach by introducing the risk operator \cite{286253}\looseness=-1
\begin{align}\label{eq:risk_op}
    \mathbb{R}_{\beta}[C]=\frac{1}{\beta} \log\left(\mathbb{E}\left[ \exp\left( \beta C \right)  \right]\right)
\end{align}
for an arbitrary random variable $C$ and risk parameter $\beta\in\mathbb{R}_+$. This operator allows the derivation of a computationally efficient condition for ensuring $\delta$-safety as shown in the following proposition.
\begin{proposition}\label{prop:safety}
Consider a cost function $c(\cdot)$ satisfying~\eqref{eq:c_cond}. If there exist constants $\xi,\beta\!\in\!\mathbb{R}_+$ with $\xi\!<\!\bar{\xi}$ such that\looseness=-1
\begin{align}\label{eq:safety_cond}
    \mathbb{R}_{\beta}[V_{\bm{\pi}}(\bm{x}^+)]\leq \xi, \qquad \forall\bm{x}\in\mathbb{V}_{\bm{\pi}}^{\bar{\xi}}
\end{align}
holds for %$\beta>0$ and  
$\bm{x}^+\!\!=\!\!\bm{f}(\bm{x},\bm{\pi}(\bm{x}),\bm{\omega})$, then, $\bm{\pi}(\cdot)$ is $\delta$-safe on $\mathbb{V}_{\bm{\pi}}^{\xi}$ with 
\begin{align}\label{eq:delta}
    \delta = \exp\left(\beta\left(\xi-\bar{\xi}\right)\right).
\end{align}
\end{proposition}
\begin{proof}
Due to \cref{lem:set2cost}, we can bound the probability of leaving $\mathbb{X}_{\mathrm{safe}}$ by the probability of leaving $\mathbb{V}_{\bm{\pi}}^{\bar{\xi}}$. Therefore, it is sufficient to 
derive an upper bound for the probability
\begin{align}\label{eq:safety_pf_1}
    \mathbb{P}\left( V_{\bm{\pi}}(\bm{x}^+)\geq \bar{\xi}\right)=\mathbb{E}_{\bm{x}^+}\left[I_{\bar{\xi}}(V_{\bm{\pi}}(\bm{x}^+)\right],
\end{align}
where the indicator function $I_{\bar{\xi}}:\mathbb{R}\rightarrow \{0,1\}$ is defined as
\begin{align}
    I_{\bar{\xi}}(V)=\begin{cases}
    0&\text{if } V\leq \bar{\xi}\\
    1&\text{if } V>\bar{\xi}.
    \end{cases}
\end{align}
Note that $V_{\bm{\pi}}(\cdot)$ is a deterministic function, such that the expectation affects only the random variable $\bm{x}^+$ in \eqref{eq:safety_pf_1}. Moreover, $\beta$ is positive, $\exp(0)=1$ and the exponential function is strictly increasing and positive. Therefore, we can bound the indicator function through the exponential expression\looseness=-1
\begin{align}
    I_{\bar{\xi}}(V_{\bm{\pi}}(\bm{x}^+))\leq \exp\left(\beta\left( V_{\bm{\pi}}(\bm{x}^+)-\bar{\xi} \right)\right)
\end{align}
due to the positivity of $\beta$. By taking the expectation of both sides, this inequality immediately leads to
\begin{align}\label{eq:safety_pf_3}
    &\!\mathbb{P}\left( V_{\bm{\pi}}(\bm{x}^+)\geq \bar{\xi}\right)\leq \mathbb{E}_{\bm{x}^+}\!\left[ \exp\left( \beta V_{\bm{\pi}}(\bm{x}^+) \right) \right]\exp(-\beta\bar{\xi}).\!
\end{align}
Due to the definition of the risk operator in \eqref{eq:risk_op}, we can simplify the right side of this inequality to obtain
\begin{align}\label{eq:safety_pf_4}
    &\mathbb{P}\left( V_{\bm{\pi}}(\bm{x}^+)\geq \bar{\xi}\right)\leq\exp\left(\beta \left( \mathbb{R}_{\beta}[V_{\bm{\pi}}(\bm{x}^+)]-\bar{\xi}\right)\right).
\end{align}
Since $\mathbb{R}_{\beta}[V_{\bm{\pi}}(\bm{x}^+)]\leq \xi$ is ensured by \eqref{eq:safety_cond}, we have 
$\mathbb{P}\left( V_{\bm{\pi}}(\bm{x}^+)\geq \bar{\xi}\right)\leq\delta$ with $\delta$ defined in \eqref{eq:delta}.
\end{proof}
This result provides a straightforward condition, which merely requires the evaluation of the risk operator and the computation of the cumulative cost, which is a problem commonly encountered in reinforcement learning. Moreover, it offers a simple expression for the probability of safety, such that it can easily be computed in practice. 

\begin{remark}\label{rem:safety_prob}
Since the probability of a safety violation~$\delta$ guaranteed by \cref{prop:safety} only depends on three parameters, it allows an intuitive interpretation:
\begin{itemize}
    \item The difference between $\xi$ and $\bar{\xi}$ can be interpreted as a safety margin since it requires the dynamics to be contractive on the set $\mathbb{V}_{\bm{\pi}}^{\bar{\xi}}\setminus\mathbb{V}_{\bm{\pi}}^{\xi}$ towards $\mathbb{V}_{\bm{\pi}}^{\xi}$. The larger this safety margin, the more contractive is the behavior at the boundary of $\mathbb{V}_{\bm{\pi}}^{\bar{\xi}}$ and consequently, it becomes more unlikely that the state reaches $\mathbb{X}\setminus\mathbb{V}_{\bm{\pi}}^{\bar{\xi}}$.
    \item The parameter $\beta$ reflects the risk-sensitivity of the safety condition \eqref{eq:safety_cond}. A large value of $\beta$ corresponds to a high risk-aversion since it causes the tails of the noise distribution $\rho$ and the function distribution $\mathcal{F}$ to have a larger effect on the left side of \eqref{eq:safety_cond}. In the extreme case of $\beta\rightarrow \infty$, this leads to \eqref{eq:safety_cond} corresponding to a condition on the worst case realization of $\bm{\omega}_k$ and $\bm{f}(\cdot)$ \cite{286253}. This increasing risk-aversion with growing~$\beta$ is intuitively accompanied by an increase in the probability of safety.
\end{itemize}
% It is important to note that these parameters cannot be chosen independently in general since the probability distributions~$\rho$ and $\mathcal{F}$ impose constraints on the admissible values for the satisfaction of \eqref{eq:safety_cond}. An increasing value of $\beta$ generally leads to greater values on the left of \eqref{eq:safety_cond}, while reducing $\xi$ tightens the necessary safety condition \eqref{eq:safety_cond} on the right side. 
\end{remark}

%%%%%%%%%%%%%%%%
\subsection{Safe backup Policies via Reinforcement Learning}\label{subsec:safe_backup}
While \cref{subsec:cons2cost} describes an approach for obtaining the probability of safety for a given policy, it does not address the problem of determining a safe policy. In this section, we show that this problem can be solved using standard reinforcement learning techniques through the following minimization problem
\begin{align}\label{eq:safe_pol}
    \bm{\pi}_{\mathrm{safe}} = \argmin\limits_{\bm{\pi}\in\Pi}\mathbb{E}_{\bm{x}}\left[ V_{\bm{\pi}}(\bm{x}) \right].
\end{align}
Even though this optimization problem does not involve the risk operator $\mathbb{R}_{\beta}[\cdot]$, its solution $\bm{\pi}_{\mathrm{safe}}$ is guaranteed to satisfy the conditions of \cref{prop:safety} under weak assumptions. This is demonstrated by the subsequent theorem. The proof follows after a discussion of the assumptions.% the focus of the remainder of this section. 
\begin{theorem}\label{th:safe_policy}
    Consider a cost function $c(\cdot)$ satisfying \eqref{eq:c_cond} and 
    assume that there exist a policy $\tilde{\bm{\pi}}(\cdot)$ and constants $\theta_1,\theta_2\in\mathbb{R}_+$ with $\theta_1<\nicefrac{1}{(1-\gamma)}$ such that
    \begin{align}\label{eq:value_bound}
        V_{\bm{\pi}}(\bm{x})\leq \theta_1 c(\bm{x})+ \theta_2, \quad\forall \bm{x}\in\mathbb{X}
    \end{align}
    is satisfied. Moreover, assume there exist constants $\theta_3,\theta_4\in\mathbb{R}_{0,+}$ such that
    \begin{align}\label{eq:V_lower_bound}
        V_{\bm{\pi}}(\bm{x})\geq \theta_3 c(\bm{x})+\theta_4, \quad\forall \bm{x}\in\mathbb{X}
    \end{align}
    holds for all policies $\bm{\pi}(\cdot)$.
    If
    \begin{align}\label{eq:c_lower_bound}
        \hat{c}>\frac{\theta_2}{\theta_3(\theta_1(\gamma-1)+1)}-\frac{\theta_4}{\theta_3}
    \end{align}
    holds,
    then, the policy \eqref{eq:safe_pol}
    is $\delta^*$-safe on $\mathbb{V}_{\xi^*}$ with $\delta^*=\exp\left(\beta^*\left(\xi^*-\bar{\xi}\right)\right)$, where
    \begin{subequations}
    \label{eq:opt_prob}
    \begin{align}
        \beta^*,\xi^* = &\argmin\limits_{\beta\in\mathbb{R}_+, \xi\in\mathbb{R}_+} \exp\left(\beta\left(\xi-\bar{\xi}\right)\right)\\
        &\text{s.t. }  \xi< \bar{\xi}\label{eq:xi_constraint}\\
        &\quad~ \text{\eqref{eq:safety_cond} holds.}
    \end{align}
    \end{subequations}
\end{theorem}
\paragraph*{Discussion}
While large values for $\theta_3$ and $\theta_4$ in \eqref{eq:V_lower_bound} are generally beneficial for admitting larger values of $\hat{c}$ in \eqref{eq:c_lower_bound}, it is always possible to trivially choose $\theta_3=1$, $\theta_4=0$ due to non-negativity of $c(\cdot)$. 
Condition \eqref{eq:value_bound} essentially requires a sufficiently fast decay of the immediate costs $c(\bm{x}_k)$ along trajectories for some policy $\tilde{\bm{\pi}}(\cdot)$. This decay can be achieved if, e.g., variants of exponential controllability hold \cite{Gaitsgory2018}. Since merely the existence of a policy $\tilde{\bm{\pi}}(\cdot)$ satisfying \eqref{eq:value_bound} is necessary, this admits the derivation of the constants $\theta_1$ and $\theta_2$ via properties such as exponential controllability \cite{Gaitsgory2018}. Therefore, the assumptions of \cref{th:safe_policy} are not restrictive in practice.\looseness=-1

Note that the required lower bound \eqref{eq:V_lower_bound} for all possible cost functions $V_{\bm{\pi}}(\cdot)$ is only necessary because of the offset $\theta_2$, which leads to a lower bound for the admissible values of $\bar{\xi}$. Since the admissible value $\bar{\xi}$ depends directly on the cost function $V_{\bm{\pi}}(\cdot)$, it indirectly depends on the policy $\bm{\pi}(\cdot)$. Therefore, $V_{\tilde{\bm{\pi}}}(\cdot)$ and $V_{\bm{\pi}_{\mathrm{safe}}}(\cdot)$ potentially admit different values for $\bar{\xi}$, such that general constraints cannot be posed on~$\bar{\xi}$. This issue is resolved by \eqref{eq:V_lower_bound}, which establishes a direct relationship between $\hat{c}$ and $\bar{\xi}$ for all possible cost functions $V_{\bm{\pi}}(\cdot)$ and thereby leads to the lower bound \eqref{eq:c_lower_bound}. If no offset exists, i.e., $\theta_2=\theta_4=0$, it can be easily seen that $\hat{c}>0$ must be satisfied. This is the trivial lower bound for $\hat{c}$ due to the assumed non-negativity of immediate cost functions $c(\cdot)$. Therefore, the offset $\theta_2$ is the only reason for the restriction of the admissible threshold $\hat{c}$.

\paragraph*{Proof} In order to prove \cref{th:safe_policy}, 
we first show that a risk-neutral variant of condition \eqref{eq:safety_cond} guarantees the existence of parameters $\xi$ and $\beta$ satisfying the requirements of \cref{prop:safety}.
\begin{lemma}\label{lem:risk2exp}
    Assume that 
    \begin{align}\label{eq:safety_cond_exp}
        \mathbb{E}_{\bm{x}^+}[V_{\bm{\pi}}(\bm{x}^+)]\leq \tilde{\xi}, \qquad \forall\bm{x}\in\mathbb{V}_{\bm{\bar{\xi}}}
    \end{align}
    holds for some constant $\tilde{\xi}<\bar{\xi}$.
    Then, there exist constants $\beta\in\mathbb{R}_+$ and $\xi<\bar{\xi}$ such that \eqref{eq:safety_cond} is satisfied.
\end{lemma}
\begin{proof}
By the Taylor series expansion of the exponential function, we have
\begin{align}
    \mathbb{R}_{\beta}[V_{\bm{\pi}}&(\bm{x}^+)]= \\
    & \frac{1}{\beta} \log \!\left(\!1 \!+\! \beta \mathbb{E}_{\bm{x}^+} \!\left[V_{\bm{\pi}}(\bm{x}^+)\right] \!+\! \frac{\beta^2}{2}\mathbb{E}_{\bm{x}^+} \!\left[V_{\bm{\pi}}^2(\bm{x}^+)\right] \!+\! \dots\!\right)\! .\nonumber
\end{align}
From the premise of the lemma, it follows that
\begin{align}
\mathbb{R}_{\beta}[V_{\bm{\pi}}(\bm{x}^+)] \leq& \\
    &\frac{1}{\beta} \log \!\left(\!1 \!+\! \beta \tilde{\xi} \!+\! \frac{\beta^2}{2}\mathbb{E}_{\bm{x}^+}\!\! \left[V_{\bm{\pi}}^2(\bm{x}^+)\right] \!+\! \dots\!\right)\!.\nonumber
\end{align}
Since $\log (1+a)<a$ for $a \in\mathbb{R}_+$ and by noting the positivity of $V_{\bm{\pi}}(\bm{x}^+)$ and the risk-aversion parameter $\beta$, we have
\begin{align}\label{eq:risk_Taylor}
    \mathbb{R}_{\beta}[V_{\bm{\pi}}(\bm{x}^+)]&< \tilde{\xi} + \beta\left(\frac{1}{2}\mathbb{E}_{\bm{x}^+} \left[V_{\bm{\pi}}^2(\bm{x}^+)\right] + \dots\right).
\end{align}
Since the second summand can be brought arbitrarily close to $0$ by choosing a sufficiently small $\beta$, there exists a $\beta$ such that the right side of \eqref{eq:risk_Taylor} is smaller than $\bar{\xi}$, which concludes the proof.
\end{proof}
The key idea behind this result is that \eqref{eq:safety_cond} converges to~\eqref{eq:safety_cond_exp} for $\beta\rightarrow 0$. Therefore, it is sufficient to determine a policy $\bm{\pi}$, which satisfies the risk-neutral condition \eqref{eq:safety_cond_exp}, for ensuring~\eqref{eq:safety_cond} with a suitably small value of $\beta\in\mathbb{R}_+$. 

Although \eqref{eq:safety_cond_exp} is a risk-neutral condition, it exhibits an expectation with respect to the next state $\bm{x}^+$. Therefore, it does not directly enable the applicability of standard RL techniques and consequently, it does not coincide with the acquisition function considered in the definition of the safe policy \eqref{eq:safe_pol}. In order to overcome this issue, we exploit \eqref{eq:value_bound} to relate $\mathbb{E}_{\bm{x}^+}[V_{\bm{\pi}}(\bm{x}^+)]$ to $V_{\bm{\pi}}(\bm{x})$. This is achieved using the following lemma.

\begin{lemma}\label{th:decrease_guarantee}
Assume that there exist $\theta_1,\theta_2\in\mathbb{R}_+$ with $\theta_1<\nicefrac{1}{(1-\gamma)}$ such that \eqref{eq:value_bound} is satisfied.
Then, it holds that
\begin{align}
    \mathbb{E}_{\bm{x}^+}\![V_{\bm{\pi}}(\bm{x}^+)]\!-\!V_{\bm{\pi}}(\bm{x}) \leq \frac{\theta_1-\theta_1\gamma-1}{\theta_1\gamma}V_{\bm{\pi}}(\bm{x})+\frac{\theta_2}{\gamma\theta_1}.
\end{align}
\end{lemma}
\begin{proof}
By solving Bellman's identity
\begin{align}
    V_{\bm{\pi}}(\bm{x})=c(\bm{x})+\gamma \mathbb{E}_{\bm{x}^+}\left[V_{\bm{\pi}}(\bm{x}')\right],
\end{align}
for $\mathbb{E}_{\bm{x}^+}\left[V_{\bm{\pi}}(\bm{x}')\right]$, 
we can express $\Delta V_{\bm{\pi}}(\bm{x})=\mathbb{E}_{\bm{x}^+}[V_{\bm{\pi}}(\bm{x}^+)]-V_{\bm{\pi}}(\bm{x})$ as
\begin{align}\label{eq:DeltaV_with c}
    \Delta V_{\bm{\pi}}(\bm{x})=\frac{1}{\gamma}(-c(\bm{x})+(1-\gamma)V_{\bm{\pi}}(\bm{x})).
\end{align}
Due to \eqref{eq:value_bound}, we have
\begin{align}
    c(\bm{x})\geq\frac{V_{\bm{\pi}}(\bm{x})-\theta_2}{\theta_1},
\end{align}
which allows us to bound \eqref{eq:DeltaV_with c} by
\begin{align}
    \!\!\Delta V_{\bm{\pi}}(\bm{x})&\leq\frac{1}{\gamma}\left(-\frac{V_{\bm{\pi}}(\bm{x})-\theta_2}{\theta_1}+(1-\gamma)V_{\bm{\pi}}(\bm{x})\right).
\end{align}
Rearranging the terms on the right side finally yields
\begin{align}
    \Delta V_{\bm{\pi}}\leq \frac{\theta_1-\theta_1\gamma-1}{\theta_1\gamma}V_{\bm{\pi}}(\bm{x})+\frac{\theta_2}{\gamma\theta_1},
\end{align}
where $\nicefrac{(\theta_1-\theta_1\gamma-1)}{\theta_1\gamma}$ is guaranteed to be negative since $\theta_1<\nicefrac{1}{(1-\gamma)}$ is assumed. 
\end{proof}

\cref{th:decrease_guarantee} ensures that the minimization of $V_{\bm{\pi}}(\bm{x})$ also reduces $\mathbb{E}_{\bm{x}^+}[V_{\bm{\pi}}(\bm{x}^+)]$. This directly allows proving  \cref{th:safe_policy} in combination with \cref{lem:risk2exp} as shown in the following.

\noindent\hspace{2em}{\itshape Proof of \cref{th:safe_policy}: }
It is straightforward to see that optimizing with respect to the expectation over $\bm{x}$ yields identical policies $\bm{\pi}_{\mathrm{safe}}(\cdot)$ as the point-wise optimum $\pi_{\bm{x}}(\bm{x})=\argmin_{\bm{\pi}\in\Pi} V_{\bm{\pi}}(\bm{x})$ for a given $\bm{x}$ and a continuous transition function $\bm{f}(\cdot,\cdot,\cdot)$. Due to optimality of $\bm{\pi}_{\bm{x}}(\cdot)$, we additionally have the inequality $V_{\bm{x}}(\bm{x})\leq V_{\tilde{\bm{\pi}}}(\bm{x})$ for all $\bm{x}\in\mathbb{X}$. Therefore, it follows from \cref{th:decrease_guarantee} that
    \begin{align}\label{eq:Vsafe_bound}
        \mathbb{E}[V_{\bm{\pi}_{\mathrm{safe}}}(\bm{x}^+)]%&\leq \left( 1\!+\!\frac{\theta_1\!-\!\theta_1\gamma\!-\!1}{\theta_1\gamma}\right)V_{\bm{\pi}_{\mathrm{safe}}}(\bm{x}) \!+\!\frac{\theta_2}{\gamma\theta_1}.\\
        &\leq \frac{1}{\gamma}\left(1-\frac{1}{\theta_1}\right)V_{\bm{\pi}_{\mathrm{safe}}}(\bm{x}) +\frac{\theta_2}{\gamma\theta_1}.
    \end{align}
    Since the right side of \eqref{eq:Vsafe_bound} is linear in $V_{\bm{\pi}_{\mathrm{safe}}}(\bm{x})$, the maximum inside $\mathbb{V}_{\bar{\xi}}$ is achieved for $V_{\bm{\pi}_{\mathrm{safe}}}(\bm{x})=\bar{\xi}$.
    Therefore, we obtain the inequality 
    \begin{align}\label{eq:xi_bound_implicit}
        \bar{\xi}>\frac{1}{\gamma}\left(1-\frac{1}{\theta_1}\right)\bar{\xi} +\frac{\theta_2}{\gamma\theta_1}
    \end{align}
    since \cref{lem:risk2exp} requires $\mathbb{E}[V_{\bm{\pi}_{\mathrm{safe}}}(\bm{x}^+)]\leq \xi<\bar{\xi}$. 
    Solving for~$\bar{\xi}$ and noting that $\bar{\xi}=\theta_3\hat{c}+\theta_4$ due to \eqref{eq:V_lower_bound} yields
    \begin{align}
        \theta_3\hat{c}+\theta_4>\frac{\theta_2}{\theta_1(\gamma-1)+1}.
    \end{align}
    \newcommand*{\QEDA}{\null\nobreak\hfill\ensuremath{\blacksquare}}
    It is straightforward to see that \eqref{eq:c_lower_bound} guarantees the satisfaction of this inequality, such that \cref{lem:risk2exp} and \cref{prop:safety} ensure that \eqref{eq:opt_prob} is feasible and results in a probability $\delta^*<1$. This immediately implies $\delta^*$-safety of $\bm{\pi}_{\mathrm{safe}}(\cdot)$ and thereby concludes the proof.\null\nobreak\hfill\QEDclosed

%%%%%%%%%%%%%%%%
\subsection{Risk-Sensitive Inhibitory Control for Safe Roll-outs}\label{subsec:safety_filt}

\begin{algorithm}[t]
\newcommand\mycommfont[1]{\footnotesize\ttfamily\textcolor{blue}{#1}}
\SetCommentSty{mycommfont}
    % \setstretch{1.5}
    \SetInd{0.55em}{0.55em}
    \DontPrintSemicolon
    % Sample $N$ functions $\hat{f}_n(\cdot)\sim\mathcal{F}$\;
    \tcc{Solve \eqref{eq:opt_pol}}
    \While{optimization not converged}{
        Sample function $\hat{f}(\cdot)\sim\mathcal{F}$\;
        Roll-out policy $\bm{\pi}^*(\cdot)$ on $\hat{f}(\cdot)$\;
        Update $\bm{\pi}^*(\cdot)$ using gathered system data\;
    }
    \tcc{Solve \eqref{eq:safe_pol}}
    \While{optimization not converged}{
        Sample function $\hat{f}(\cdot)\sim\mathcal{F}$\;
        Roll-out policy $\bm{\pi}_{\mathrm{safe}}(\cdot)$ on $\hat{f}(\cdot)$\;
        Update $\bm{\pi}_{\mathrm{safe}}(\cdot)$ using gathered system data\;
    }
    \tcc{$\!$Safe roll-out via online optimization~\eqref{eq:safety_filt_fixed_xi}$\!\!$}
    Apply $\bm{\pi}_{\mathrm{safe}}^*(\cdot)$ to unknown system $f(\cdot)$\;
	\caption{Safe RL using Risk-Sensitive Filters}
	\label{alg:safe_RL}
\end{algorithm}

Based on the safe policy $\bm{\pi}_{\mathrm{safe}}(\cdot)$ obtained using \eqref{eq:safe_pol}, we propose a risk-sensitive inhibitory control strategy for enabling safe RL as outlined in Alg.~\ref{alg:safe_RL}. For this purpose, we first obtain an optimal, potentially unsafe policy by solving the optimization problem
\begin{align}\label{eq:opt_pol}
    \bm{\pi}^*=\argmax\limits_{\bm{\pi}\in\Pi} \mathbb{E}_{\bm{f},\bm{\omega},\bm{x}_0}\left[\sum\limits_{k=0}^{\infty} \gamma^k r(\bm{x}_k,\bm{\pi}(\bm{x}_k))\right],
\end{align}
where $r:\mathbb{X}\times\mathbb{U}\rightarrow\mathbb{R}_{0,+}$ denotes a reward function and $\bm{x}_k$ is defined through the iterative application of \eqref{eq:true_sys} with $\bm{x}_0=\bm{x}$ and $\bm{u}_k=\bm{\pi}(\bm{x}_k)$. This problem can be solved using standard off-policy reinforcement learning algorithms such as soft actor-critic reinforcement learning \cite{haarnoja2018soft}. 
Afterward, a safe backup policy $\bm{\pi}_{\mathrm{safe}}(\cdot)$ is computed by solving \eqref{eq:safe_pol}, which can be straightforwardly achieved using standard off-policy reinforcement learning techniques. Finally, we apply the policy to the true system \eqref{eq:true_sys}. For this roll-out, we employ the risk-sensitive filter
\begin{subequations}\label{eq:safety_filt_fixed_xi}
\begin{align}
    \bm{\pi}^*_{\mathrm{safe}}(\bm{x})=&\argmin\limits_{\bm{u}\in\mathbb{U}}\|\bm{\pi}^*(\bm{x})-\bm{u}\|\\
    &\text{s.t. }\mathbb{R}_{\beta}[V_{\bm{\pi}_{\mathrm{safe}}}(\bm{f}(\bm{x},\bm{u},\bm{\omega}))]\leq \xi^*
    \label{eq:safety_filt_const}
\end{align}
\end{subequations}
which makes use of the safe backup policy $\bm{\pi}_{\mathrm{safe}}(\cdot)$ through the cost function $V_{\bm{\pi}_{\mathrm{safe}}}$ and minimally adjusts the policy $\bm{\pi}^*(\cdot)$ such that the safety condition \eqref{eq:safety_cond} is satisfied. 

Due to the safety filter \eqref{eq:safety_filt_fixed_xi}, the state constraints $\mathbb{X}_{\mathrm{safe}}$ can straightforwardly be considered in Alg.~\ref{alg:safe_RL}. In fact, $\delta$-safety of $\bm{\pi}^*_{\mathrm{safe}}(\cdot)$ is directly inherited from the safe backup policy $\bm{\pi}_{\mathrm{safe}}(\cdot)$ as shown in the following theorem.
\begin{theorem}\label{th:safety_filt}
    Consider a cost function $c(\cdot)$ satisfying \eqref{eq:c_cond} and a threshold $\hat{c}$, for which \eqref{eq:c_lower_bound} holds. Moreover, assume that there exists a policy $\tilde{\bm{\pi}}(\cdot)$ satisfying \eqref{eq:value_bound} with $\theta_1<\nicefrac{1}{(1-\gamma)}$ for all $\bm{x}\in\mathbb{X}_{\mathrm{safe}}$. Then, the safety filtered policy \eqref{eq:safety_filt_fixed_xi} is $\delta^*$-safe on $\mathbb{V}_{\bm{\pi}_{\mathrm{safe}}}^{\xi^*}$ with $\delta^*=\exp\left(\beta^*\left(\xi^*-\bar{\xi}\right)\right)$, where $\beta^*$ and $\xi^*$ are defined in \eqref{eq:opt_prob}. 
\end{theorem}
\begin{proof}
    Due to \cref{th:safe_policy}, $\bm{\pi}_{\mathrm{safe}}(\cdot)$ defined in \eqref{eq:safe_pol} satisfies \eqref{eq:safety_filt_const}. Thus, the optimization problem \eqref{eq:safety_filt_fixed_xi} is guaranteed to be feasible for all states $\bm{x}\in\mathbb{V}_{\bm{\pi}_{\mathrm{safe}}}^{\xi^*}$ with the trivial solution $\bm{u}=\bm{\pi}_{\mathrm{safe}}(\bm{x})$. Finally, $\delta^*$-safety directly follows from \cref{prop:safety}.
\end{proof}
While this theorem employs the optimal parameters $\beta^*$ and $\xi^*$, it immediately follows from the proof of \cref{th:safe_policy} that for every value $\xi$ with $\xi^*\leq \xi<\bar{\xi}$, there exists a $\beta\in\mathbb{R}_+$ satisfying \eqref{eq:xi_constraint}. Therefore, $\delta$-safety on $\mathbb{V}_{\xi} \supset\mathbb{V}_{\xi^*}$ with $\delta>\delta^*$ can be straightforwardly ensured in practice by choosing a sufficiently large value $\xi<\bar{\xi}$ and a suitably small value $\beta\in\mathbb{R}_+$. 

\begin{remark}
    When $\beta$ becomes larger, the control becomes more pessimistic, and therefore, the probability of safety generally increases. However, there exists a critical value at which the safety constraint \eqref{eq:safety_filt_const} becomes infeasible for all $\xi<\bar{\xi}$. That is, the control becomes too phobic to act. This resembles a well-known behavior in risk-sensitive control and RL commonly referred to as neurotic breakdown \cite{fleming2006risk}. 
\end{remark}

\section{Simulations}\label{sec:num_eval}        

In this section, we evaluate the proposed risk-sensitive inhibitory control approach, described in Alg.~\ref{alg:safe_RL}, using the popular Mujoco Half-Cheetah environment \cite{todorov2012mujoco}. The Half-Cheetah is a planar model of a large, cat-like robot with 6 actuated joints. The main goal is to maximize the robot's walking velocity with the least control effort possible, which is encoded in the default reward function.
We consider the default model parameters for the Cheetah robot, but assume a body mass perturbed by a Gaussian distributed random variable with $0$ mean and standard deviation $0.1$. In order to obtain a challenging safety condition, we 
set optimality and safety in a direct conflict similar as in \cite{Curi2022} by 
%follow the problem in \cite{Curi2022} and set it in direct conflict with the optimality goal by 
constraining the velocity to $v\leq v_{\mathrm{crit}}$, $v_{\mathrm{crit}}=2$. 
As cost function for the computation of the safe policy \eqref{eq:safe_pol}, $c(\bm{x})=v-\underline{v}$ is employed with threshold $\hat{c}=2-\underline{v}$, where $\underline{v}=\changed{-}10$ denotes the considered minimum velocity of the Half-Cheetah robot. This cost function encourages the robot to run with a negative velocity, such that the distance to the safety threshold velocity $v_{\mathrm{crit}}$ is maximized. Note that the subtraction of $\underline{v}$ is necessary to ensure the non-negativity of the cost $c(\cdot)$ assumed in our derivations, but it merely causes a constant off-set in the cumulative cost $V_{\bm{\pi}}(\cdot)$.

The optimal and safe policies are obtained using the Soft-Actor Critic (SAC) algorithm \cite{haarnoja2018soft} with $400$ training iterations each with $1000$ time steps and the hyper-parameters provided by \cite{Liang2018}. For computing the expectations over dynamics $\bm{f}(\cdot)$ in \eqref{eq:cum_cost} and \eqref{eq:opt_pol}, we randomly sample $10$ body masses, such that we can use the corresponding sample environments to empirically approximate all necessary expected values.
The risk-sensitive safety filter \eqref{eq:safety_filt_fixed_xi} is implemented using the cross-entropy method~\cite{botev2013cross} with $5$ iterations per time step and $10$ particles. The safety constraints are considered in an augmented objective function using fixed Lagrange multipliers, such that they are effectively enforced using soft constraints to allow recovery after constraint violations. The risk operator $\mathbb{R}_{\beta}[\cdot]$ is approximated through $100$ sample environments. For each parameter combination $(\xi,\beta)$, $100$ time steps are simulated and $3$ random seeds are averaged.%\looseness=-1

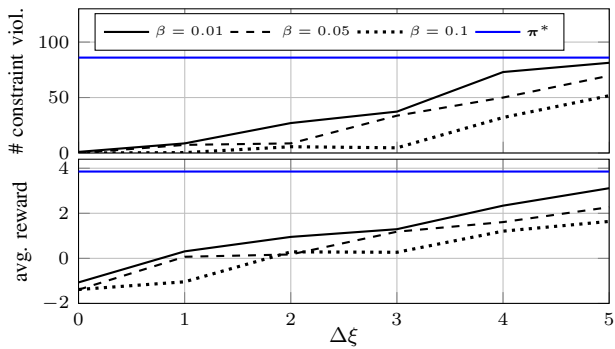
\begin{figure}
    \centering
    \def\fileopt{eval_test.txt}
    \def\file{sim_results.txt}
    % \tikzsetnextfilename{e_T_plot}
	\begin{tikzpicture}
	\begin{axis}[ylabel={\footnotesize \# constraint viol.},
	xmin=0.0, ymin = 0, xmax = 5, ymax = 130, height =3.5cm,
	x label style={yshift=0.2cm}, y label style={yshift=-0.05cm}, 
	legend pos=north west, legend columns=4, name=plot1,xticklabels={,,}]
	\addplot[black,thick]	table[x = zetas_001, y  = const_state_001]{\file};
    \addplot[black,thick,dashed]	table[x = zetas_005, y  = const_state_005]{\file};
    \addplot[black,very thick,dotted]	table[x = zetas_01, y  = const_state_01]{\file};
    \addplot[blue,thick]	table[x = zeta, y  = state_viol_opt]{\fileopt};
	\legend{$\beta=0.01$, $\beta=0.05$, $\beta=0.1$, $\bm{\pi}^*$};
	\end{axis}
 %    \begin{axis}[ylabel={\# infeasibilities},
	% xmin=0.0, ymin = 0.0, xmax = 5, ymax = 29, height =4cm,
	% x label style={yshift=0.2cm}, y label style={yshift=-0.05cm}, 
	% legend pos=south west, name=plot1,xticklabels={,,},
 %    at=(plot1.south), anchor=south,
	% yshift=-2.5cm, name=plot2,]
	% \addplot[black,thick]	table[x = zeta, y  = cost_viol]{\file};
 %    \addplot[black,thick, dashed]	table[x = zeta, y  = cost_viol_01]{\file};
	% % \legend{};
	% \end{axis}
 %    \begin{axis}[ylabel={avg. velocity},
	% xmin=0.0, ymin = -0.5, xmax = 5, ymax = 3.5, height =4cm,
	% x label style={yshift=0.2cm}, y label style={yshift=-0.05cm}, 
	% legend pos=south west, name=plot1,xticklabels={,,},
 %    at=(plot2.south), anchor=south,
	% yshift=-2.5cm, name=plot3,]
	% \addplot[black,thick]	table[x = zeta, y  = velocity]{\file};
 %    \addplot[black,thick,dashed]	table[x = zeta, y  = velocity_01]{\file};
	% % \legend{};
	% \end{axis}
	\begin{axis}
	[xlabel={\footnotesize $\Delta \xi$},ylabel={\footnotesize avg. reward},
	xmin=0.0, ymin = -2, xmax = 5, ymax = 4.4, height =3.5cm,
    y label style={yshift=-0.1cm}, x label style={yshift=0.2cm}, 
	legend pos=south west, at=(plot1.south), anchor=south,
	yshift=-2.0cm]
	\addplot[black,thick]	table[x = zetas_001, y  = rew_001]{\file};
    \addplot[black,thick,dashed]	table[x = zetas_001, y  = rew_005]{\file};
    \addplot[black,very thick,dotted]	table[x = zetas_001, y  = rew_01]{\file};
    \addplot[blue,thick]	table[x = zeta, y  = reward_opt]{\fileopt};
	\end{axis}
	\end{tikzpicture}
    \vspace{-0.35cm}
    \caption{Number of constraint violations and average rewards in dependency on the safety constraint threshold $\xi=521+\Delta \xi$ and the risk-sensitivity $\beta$. Reducing $\beta$ and increasing $\xi$ have a similar effect of admitting more risky behavior in the response inhibition, such that the number of constraint violations and the average reward increase.}
    \label{fig:filter_performance}
\end{figure}

The resulting numbers of constraint violations and the average reward for different values of $\beta$ and $\xi$ are depicted in \cref{fig:filter_performance}. We can observe that 
increasing $\xi$ has exactly the expected effect of loosening the safety constraint by admitting higher velocities $v$, such that the probability of safety decreases and more constraint violations can be observed. At the same time, this allows a higher robot velocity, which in turn causes an increasing average reward. A similar effect can be observed with the risk parameter $\beta$ due to the considered state-independent model uncertainty. When $\beta$ is increased, the conservatism of the safety filter increases. This leads to a lower number of constraint violations, but the average reward also reduces. Therefore, the parameters $\xi$ and $\beta$ exhibit the impact on the probability of safety as discussed in \cref{rem:safety_prob}. Note that the risk-inhibition with the considered soft constraint formulation has a clearly visible effect on the average robot velocity, even when it does not manage to enforce the safety constraints. This can be observed in a comparison with the optimal policy $\bm{\pi}^*(\cdot)$, which achieves a significantly higher reward with a similar number of constraint violations for large values of $\xi$ and small $\beta$. Therefore, the proposed risk-sensitive inhibitory control not only allows to reduce the number of constraint violations, but also the amount by which the constraint is violated.\looseness=-1

\section{Conclusion}\label{sec:conclusion}
Inspired by the psychological concept of inhibitory control, this paper proposes a risk-sensitive method for rendering arbitrary policies safe. This method is based on the introduction of cost functions, such that state constraints can be expressed in terms of value functions. We show that this formulation allows us to employ standard reinforcement learning techniques for obtaining policies that their only goal is to ensure safety. Based on the determined safe policies and corresponding value functions, a risk-sensitive safety constraint is employed to enforce the satisfaction of state constraints online. Thereby, risk-sensitive inhibitory control is realized and its effectiveness is demonstrated in simulations.\looseness=-1


%%%%%%%%%%%%%%%%%%%%%%%%%%%%%%%%%%%%%%%%%%%%%%%%%%%%%%%%%%%%%%%%%%%%%%%%%%%%%%%%

\bibliographystyle{IEEEtran}
\bibliography{IEEEabrv,myBib.bib}

\begin{thebibliography}{10}
\providecommand{\url}[1]{#1}
\csname url@samestyle\endcsname
\providecommand{\newblock}{\relax}
\providecommand{\bibinfo}[2]{#2}
\providecommand{\BIBentrySTDinterwordspacing}{\spaceskip=0pt\relax}
\providecommand{\BIBentryALTinterwordstretchfactor}{4}
\providecommand{\BIBentryALTinterwordspacing}{\spaceskip=\fontdimen2\font plus
\BIBentryALTinterwordstretchfactor\fontdimen3\font minus \fontdimen4\font\relax}
\providecommand{\BIBforeignlanguage}[2]{{%
\expandafter\ifx\csname l@#1\endcsname\relax
\typeout{** WARNING: IEEEtran.bst: No hyphenation pattern has been}%
\typeout{** loaded for the language `#1'. Using the pattern for}%
\typeout{** the default language instead.}%
\else
\language=\csname l@#1\endcsname
\fi
#2}}
\providecommand{\BIBdecl}{\relax}
\BIBdecl

\bibitem{Nigg2000}
J.~T. Nigg, ``{On Inhibition/Disinhibition in Developmental Psychopathology: Views from Cognitive and Personality Psychology and a Working Inhibition Taxonomy},'' \emph{Psychological Bulletin}, vol. 126, no.~2, pp. 220--246, 2000.

\bibitem{Brunke2021}
L.~Brunke, M.~Greeff, A.~W. Hall, Z.~Yuan, S.~Zhou, J.~Panerati, and A.~P. Schoellig, ``{Safe Learning in Robotics: From Learning-Based Control to Safe Reinforcement Learning},'' \emph{Annual Review of Control, Robotics, and Autonomous Systems}, vol.~5, pp. 411--444, 2022.

\bibitem{Sutton2017}
R.~S. Sutton and A.~G. Barto, \emph{{Reinforcement Learning: An Introduction}}, 2nd~ed.\hskip 1em plus 0.5em minus 0.4em\relax The MIT Press, 2017.

\bibitem{Dulac-Arnold2019}
\BIBentryALTinterwordspacing
G.~Dulac-Arnold, D.~Mankowitz, and T.~Hester, ``{Challenges of Real-World Reinforcement Learning},'' in \emph{ICML Workshop on Real-Life Reinforcement Learning}, 2019. [Online]. Available: \url{http://arxiv.org/abs/1904.12901}
\BIBentrySTDinterwordspacing

\bibitem{Alshiekh2018}
M.~Alshiekh, R.~Bloem, R.~Ehlers, B.~K{\"{o}}nigshofer, S.~Niekum, and U.~Topcu, ``{Safe Reinforcement Learning via Shielding},'' in \emph{AAAI Conference on Artificial Intelligence}, 2018, pp. 2669--2678.

\bibitem{Taylor2019}
A.~Taylor, A.~Singletary, Y.~Yue, and A.~Ames, ``{Learning for Safety-Critical Control with Control Barrier Functions},'' in \emph{Learning for Dynamics {\&} Control}, 2019, pp. 708--717.

\bibitem{Bastani2021}
O.~Bastani, ``{Safe Reinforcement Learning with Nonlinear Dynamics via Model Predictive Shielding},'' in \emph{American Control Conference}, 2021, pp. 3488--3494.

\bibitem{Wabersich2021b}
K.~P. Wabersich, L.~Hewing, A.~Carron, and M.~N. Zeilinger, ``{Probabilistic Model Predictive Safety Certification for Learning-Based Control},'' \emph{IEEE Transactions on Automatic Control}, vol.~76, no.~1, pp. 176--188, 2021.

\bibitem{Hsu2021}
K.~C. Hsu, V.~Rubies-Royo, C.~J. Tomlin, and J.~F. Fisac, ``{Safety and Liveness Guarantees through Reach-Avoid Reinforcement Learning},'' in \emph{Robotics: Science and Systems}, 2021.

\bibitem{Curi2022}
S.~Curi, A.~Lederer, S.~Hirche, and A.~Krause, ``{Safe Reinforcement Learning via Confidence-Based Filters},'' in \emph{IEEE Conference on Decision and Control}, 2022.

\bibitem{sherman2018connecting}
L.~Sherman, L.~Steinberg, and J.~Chein, ``{Connecting Brain Responsivity and Real-World Risk Taking: Strengths and Limitations of Current Methodological Approaches},'' \emph{Developmental Cognitive Neuroscience}, vol.~33, pp. 27--41, 2018.

\bibitem{Ahmadi2022}
M.~Ahmadi, X.~Xiong, and A.~D. Ames, ``{Risk-Averse Control via CVaR Barrier Functions: Application to Bipedal Robot Locomotion},'' \emph{IEEE Control Systems Letters}, vol.~6, pp. 878--883, 2022.

\bibitem{Rasmussen2006}
C.~E. Rasmussen and C.~K.~I. Williams, \emph{{Gaussian Processes for Machine Learning}}.\hskip 1em plus 0.5em minus 0.4em\relax Cambridge, MA: The MIT Press, 2006.

\bibitem{Lakshminarayanan2017}
B.~Lakshminarayanan, A.~Pritzel, and C.~Blundell, ``{Simple and Scalable Predictive Uncertainty Estimation using Deep Ensembles},'' in \emph{Advances in Neural Information Processing Systems}, 2017, pp. 6405--6416.

\bibitem{286253}
M.~James, J.~Baras, and R.~Elliott, ``{Risk-Sensitive Control and Dynamic Games for Partially Observed Discrete-Time Nonlinear Systems},'' \emph{IEEE Transactions on Automatic Control}, vol.~39, no.~4, pp. 780--792, 1994.

\bibitem{Gaitsgory2018}
V.~Gaitsgory, L.~Gr{\"{u}}ne, M.~H{\"{o}}ger, C.~M. Kellett, and S.~R. Weller, ``{Stabilization of Strictly Dissipative Discrete Time Systems with Discounted Optimal Control},'' \emph{Automatica}, vol.~93, pp. 311--320, 2018.

\bibitem{haarnoja2018soft}
T.~Haarnoja, A.~Zhou, P.~Abbeel, and S.~Levine, ``{Soft Actor-Critic: Off-Policy Maximum Entropy Deep Reinforcement Learning with a Stochastic Actor},'' in \emph{International Conference on Machine Learning}, 2018, pp. 1861--1870.

\bibitem{fleming2006risk}
W.~H. Fleming, ``{Risk Sensitive Stochastic Control and Differential Games},'' \emph{Communications in Information and Systems}, vol.~6, no.~3, pp. 161--177, 2006.

\bibitem{todorov2012mujoco}
E.~Todorov, T.~Erez, and Y.~Tassa, ``Mujoco: A physics engine for model-based control,'' in \emph{IEEE/RSJ International Conference on Intelligent Robots and Systems}, 2012, pp. 5026--5033.

\bibitem{Liang2018}
E.~Liang, R.~Liaw, P.~Moritz, R.~Nishihara, R.~Fox, K.~Goldberg, J.~E. Gonzalez, M.~I. Jordan, and I.~Stoica, ``{RLlib: Abstractions for Distributed Reinforcement Learning},'' in \emph{International Conference on Machine Learning}, 2018, pp. 4768--4780.

\bibitem{botev2013cross}
Z.~I. Botev, D.~P. Kroese, R.~Y. Rubinstein, and P.~L’Ecuyer, ``The cross-entropy method for optimization,'' in \emph{Handbook of Statistics}.\hskip 1em plus 0.5em minus 0.4em\relax Elsevier, 2013, vol.~31, pp. 35--59.

\end{thebibliography}

\end{document}